\documentclass[12pt,a4paper]{article}
\usepackage{amsmath,amsthm}
\numberwithin{table}{section}
\numberwithin{figure}{section}
\usepackage{algorithm}
\usepackage{algorithmic}
\usepackage{makecell}
\usepackage{multirow}
\usepackage{tikz}
\usepackage{setspace}
\usepackage{fullpage} 
\makeatletter
\newenvironment{breakablealgorithm}
  {
   \begin{center}
     \refstepcounter{algorithm}
     \hrule height.8pt depth0pt \kern2pt
     \renewcommand{\caption}[2][\relax]{
       {\raggedright\textbf{\ALG@name~\thealgorithm} ##2\par}%
       \ifx\relax##1\relax 
         \addcontentsline{loa}{algorithm}{\protect\numberline{\thealgorithm}##2}%
       \else 
         \addcontentsline{loa}{algorithm}{\protect\numberline{\thealgorithm}##1}%
       \fi
       \kern2pt\hrule\kern2pt
     }
  }{
     \kern2pt\hrule\relax
   \end{center}
  }
\makeatother
\allowdisplaybreaks[4]
\usepackage{multicol}
\usepackage{hyperref}
\usepackage{caption}[=v1]
\usepackage{color}
\usepackage[titletoc]{appendix}
\usepackage{float}
\usepackage{bm}
\usepackage{enumerate}
\usepackage{booktabs}
\usepackage{extarrows}
\usepackage{graphicx}
\usepackage{subfigure}
\usepackage{amssymb}
\usepackage{fancyhdr}
\usepackage{authblk}
\usepackage{dsfont}
\usepackage{nonfloat}
\usepackage{xcolor}
\usepackage{tikz} 
\usetikzlibrary{arrows,shapes,chains}
\usepackage{hyperref} 
\usepackage{xfrac} 
\usepackage{longtable}
\usepackage{booktabs}
\usepackage{caption}
\usepackage[comma,authoryear]{natbib}
\newtheorem{theorem}{Theorem}[section]
\newtheorem{lemma}{Lemma}[section]

\begin{document}

\title{\textbf {The mean-variance portfolio selection based on the average and current profitability of the risky asset}\footnotemark[1]}

	\author[a]{Yu Li}
	\author[a]{Yuhan Wu}
    	\author[a,b]{Shuhua Zhang}
	\affil[a]{Coordinated Innovation Center for Computable Modeling in Management Science, Tianjin University of Finance and Economics, Tianjin 300222, China}
	\affil[b]{Zhujiang College, South China Agricultural University, Guangzhou 510900, China}
	\renewcommand*{\Affilfont}{\footnotesize\it}

\renewcommand{\thefootnote}{}
\footnotetext{E-mail addresses: liyu@tjufe.edu.cn (Yu Li), wuyuhan@stu.tjufe.edu.cn (Yuhan Wu), szhang@tjufe.edu.cn (Shuhua Zhang)}
\footnotetext[1]{This project was supported in part by the National Natural Science Foundation of China (12271395), the Humanities and Social Science Research Program of the Ministry of Education of China (22YJAZH156) and the Innovation Team Project for Ordinary University in Guangdong Province, China (2023WCXTD022).}
\date{}
\maketitle


\begin{abstract}

We study the continuous-time pre-commitment mean-variance portfolio selection in a time-varying financial market.
By introducing two indexes which respectively express the average profitability of the risky asset (AP) and the current profitability of the risky asset (CP), the optimal portfolio selection is represented by AP and CP.
Furthermore, instead of the traditional maximum likelihood estimation (MLE) of return rate and volatility of the risky asset, we estimate AP and CP with the second-order variation of an auxiliary wealth process.
We prove that the estimations of AP and CP in this paper are more accurate than that in MLE.
And, the portfolio selection is implemented in various simulated and real financial markets.
Numerical studies confirm the superior performance of our portfolio selection with the estimation of AP and CP under various evaluation criteria.

\vskip0.3cm {\bf Key words.}
Mean-variance analysis;
Parameter estimation;
The average profitability of risky asset;
The current profitability of risky asset.

\end{abstract}

\renewcommand{\thefootnote}{}
\footnotetext{Abbreviation statement: the average profitability of the risky asset (AP); the current profitability of the risky asset (CP); maximum likelihood estimation (MLE)}


\newpage
\section{Introduction}

The portfolio selection problem is the study of seeking an allocation of wealth among different risky assets in order to increase returns and diversify risks.
The mean-variance analysis proposed by \cite{Markowitz1952Portfolio} describes the price of risky assets as random variables, 
in which the expectation and variance of the portfolio are considered as investment return and investment risk, respectively.
Originally cast in a single-period framework, the mean-variance analysis has undoubtedly also inspired the development of the multi-period model \citep{Li2000Optimal} and the continuous-time model \citep{Zhou2000Continuous}.

In the continuous-time model, portfolio selections with mean-variance criterion are always considered under the assumption that the price of risky asset follows Geometric Brownian Motion (GBM) \citep{Black1973Pricing}.
And different assumptions about GBM are developed.
For example, \cite{Zhou2000Continuous} considers that the return rate and volatility of risky asset are constant.
\cite{Zhou2002Continuous} assumes the return rate and volatility in GBM are switched in limited states.
\cite{Lim2002Meanvariance} extends the assumption about the return rate and volatility of risky asset to bounded stochastic processes.

The traditional paradigm for implementing the mean-variance portfolio selection follows ``separation principle", which separates the steps between estimation and optimization.
In the first step, the model parameters are estimated from historical data using statistical methods. 
In the second step, the estimated model is taken as given, and the optimization of mean-variance analysis is focused on. 
Due to the time-inconsistency of mean-variance operator, the portfolio selection is always considered under the pre-commitment condition \citep{Zhou2000Continuous, Li2002Dynamic, Bielecki2005Continuous, Lim2002Meanvariance, Zhou2003Markowitzs}.
However, the $t$-time pre-commitment mean-variance portfolio selection is related to the return rate and volatility of risky asset in the future, which are difficult to estimate accurately. 
Even worse, the portfolio selection is always sensitive to model parameters \citep{MJ1991Sensitivity}. 
As a result, estimation errors may be significantly amplified, rendering the investment result irrelevant to investment target.

A simple way to estimate model parameters in the future is to assume that the return rate and volatility of risky asset are constant.
When implementing the portfolio selection, at each executable time point, the value of parameters are estimated with historical data
and the estimations are considered to remain the same values during the whole planning investment horizon.
However, in the financial market, volatility smile \citep{Rubinstein1994Implied} and heavy tails of the return distribution \citep{Mandelbrot1963variation} are increasingly widely observed.
These phenomena imply that the financial market is most likely non-stationary, defeating the constant assumption of the model parameters.

Over the past few decades, various stochastic volatility models have been developed based on GBM to describe the process of risky asset,
e.g., \cite{Heston1993closed}, \cite{Hull1987pricing}, \cite{Lewis2000Option} and \cite{Stein1991Stock}.
A stochastic volatility model consists of the stochastic process of risky asset price and the stochastic process of exogenous variable, in which the exogenous variable influences the return rate and volatility of risky asset.
For example, \cite{Yan2019Open} considers the mean-variance portfolio selection under the Heston model.
It assumes the return rate and volatility of risky asset in the future will follow a specific process, and all model parameters are estimated with historical data.
\cite{Zhang2021Dynamic} considers the mean-variance portfolio selection under the 3/2 stochastic volatility model.
And, the processes of risky asset price and exogenous variable are assumed to be affected by the same standard Brownian motion.
However, a model misspecification may exist, and the exogenous variable is often difficult to select or observe in real financial market.

The first contribution of this paper is the theoretical development of two indexes in optimization, which present the average profitability of the risky asset (AP) and the current profitability of the risky asset (CP).
Without assuming any knowledge about the underlying forms of the return rate and volatility of risky asset, the future information of financial market can be separated from the portfolio selection by AP.
And, the $t$-time optimal pre-commitment mean-variance portfolio selection can be represented by AP during the whole planning investment horizon and CP at current time $t$. 
We find that, as AP increases, the pre-commitment mean-variance portfolio selection will invest more in risky asset with exponential growth.

In parameter estimation, we focus on AP and CP instead of the return rate and volatility of risky asset.
For AP, we propose an auxiliary wealth process which corresponds to time-consistent mean-variance portfolio selection.
We prove that AP is equivalent to the second-order variation of the auxiliary wealth process which can be estimated by the second-order variation in the past.
Compared with the traditional method, which roughly considers the return rate and volatility of risky asset are constant and calculates AP with the square of risk premium, our estimation is more accurate. 
In fact, our estimation not only considers the time dependence characteristics of risk premium but also implicitly modifies the estimation of AP with the true value of volatility of risky asset in the calculation of second-order variation.
For CP, we approximate CP with the $t$-time estimation of AP.
The variance of our estimation of CP is less than that of the traditional method, which calculates CP with the maximum likelihood estimation (MLE) of the return rate and volatility of risky asset.

We design an algorithm to implement the pre-commitment mean-variance portfolio selection with the estimation of AP and CP. 
In various simulated and real financial markets, the superiority of our portfolio selection over the traditional estimating is confirmed.
Further, we compare our portfolio selection with the so-called ``buy-and-hold'' strategy, which invests all the money into the risky asset at the initial time 0 and makes no adjustments for optimization in $[0,T]$.
Our algorithm also performs better in nearly all financial markets under various evaluation criteria.

In summary, the main contributions of this paper are threefold.

1) We propose two indexes, which present AP and CP.

2) Through AP and CP, we separate the unknown future information of financial market from the pre-commitment mean-variance portfolio selection. 

3) We estimate AP with the second-order variation of an auxiliary wealth process, which is more accurate than MLE.
And CP is approximated by the estimation of AP, which reduces estimating volatility for CP.

The remainder of this paper is organized as follows.
In Section \ref{section:2}, we formulate the mean-variance portfolio selection problem under GBM with time-varying return rate and volatility of risky asset and derive the optimal pre-commitment mean-variance portfolio selection.
The indexes AP and CP are introduced in Section \ref{section:3}.
In Section \ref{section:4}, an algorithm is designed for our portfolio selection with the estimations of AP and CP.
In Section \ref{section:5}, our portfolio selection is implemented in various simulated and real financial markets.
Finally, we conclude in Section \ref{section:6}.

\section{Problem Formulation}\label{section:2}

In this section, we formulate the continuous-time mean-variance portfolio selection problem.
Because of the non-smoothness of the variance operator, the iterated-expectations property for mean-variance objective fails, 
and the optimal portfolio selection is given under the pre-commitment condition \citep{Li2000Optimal} .

\subsection{Continuous-Time Mean-Variance Problem}

To simplify the discussion, we consider the portfolio selection consisting of one risky asset (a well-diversified stock index) and one risk-free asset (a bond). 
This simplification allows us to focus on the primary investment issue of the risky versus risk-free mix of the portfolio, instead of secondary issues such as the composition of risky asset basket \citep{Staden2021Surprising}.

Let the planning investment horizon $[0,T]$ be fixed.
$\{B_t,0 \leqslant t \leqslant T\}$ is a standard one-dimensional Brownian motion defined on a filtered probability space $(\Omega,\mathcal{F},\{\mathcal{F}_t\}_{0 \leqslant t \leqslant T},\mathbb{P})$.
Let $r$ denote the risk-free interest rate.
Then, the discounted price of risky asset is observable, whose dynamic is governed by stochastic differential equation:
\begin{equation}\label{equ:S}
	dS_t=(\mu(t)-r) S_tdt+\sigma(t) S_tdB_t, \quad S_0=s^o,
\end{equation}
where the instantaneous return rate $\mu(t)$ and volatility $\sigma(t)$ are both time-varying, $0\leqslant t \leqslant T$.
We consider $\mu(t)$ and $\sigma(t)$ are deterministic and continuous functions of ~$t$.
However, the exact forms of $\mu(\cdot)$ and $\sigma(\cdot)$ are unknown to the investor.

Denote the investor's action by $\{\theta_t\}_{0 \leqslant t \leqslant T}$, with $\theta_t$ representing the discounted amount invested in risky asset at time $t$.
Under the self-financing assumption, the process of discounted wealth $\{W_t\}_{0 \leqslant t \leqslant T}$ can be shown as
\begin{equation}\label{equ:W_CL}
	dW_t=\theta_t(\mu(t)-r) dt +\theta_t\sigma(t) dB_t, \quad  W_0=w^o,
\end{equation}
when the financial market allows short-selling and leverage without extra cost.

The continuous-time mean-variance problem \citep{Zhang2018Portfolio} aims to consider the portfolio selection under a trade-off between the expectation and variance of the discounted terminal wealth $W_T$:
\begin{equation}\label{model:2}
	\max_{\{\theta_t\}}\quad {\rm E}\Big(W_T\Big)-\gamma {\rm Var}\Big(W_T\Big),
\end{equation}
where $\{W_t\}_{0 \leqslant t \leqslant T}$ satisfies the process \eqref{equ:W_CL} under portfolio selection $\{\theta_t\}_{0 \leqslant t \leqslant T}$.
And, $\gamma ~(\gamma>0)$ is the risk aversion coefficient, 
which is the weight between investment return and investment return.

\subsection{The Optimal Mean-Variance Portfolio Selection}

Because the variance operator is non-smooth, i.e.,
\begin{equation*}
	{\rm Var}_s \Big({\rm Var}_t \Big(\cdot\Big)\Big)\neq {\rm Var}_s\Big(\cdot\Big), \quad 0\leqslant s<t\leqslant T,
\end{equation*}
problem \eqref{model:2} is known to be time inconsistent, and the dynamic programming principle \citep{Bellman1957Dynamic} fails.
We focus on the pre-commitment portfolio selection which decides the portfolio selection $\{\theta_t\}_{0 \leqslant t \leqslant T}$ at initial time 0 and commits not to deviate from it.
Following \cite{Zhou2000Continuous}, problem \eqref{model:2} can be embedded into a tractable auxiliary problem,
and Lemma \ref{th:aux} is obtained.

\begin{lemma}[\cite{Zhou2000Continuous}]\label{th:aux}
Suppose $\{\theta^*(t)\}_{0 \leqslant t \leqslant T}$~ is the optimal discounted portfolio selection in problem \eqref{model:2}. 
Then, $\{\theta^*(t)\}_{0 \leqslant t \leqslant T}$ is also the optimal discounted portfolio selection in tractable problem \eqref{model:3}:
\begin{equation}\label{model:3}
	\max_{\{\theta_t\}}\quad {\rm E}\Big(-\gamma W_T^2+\tau W_T\Big),
\end{equation}
with the process \eqref{equ:W_CL},
in which $\tau=1+2\gamma {\rm E}\Big(W_T^*\Big)$ and $W_T^*$ is the discounted wealth at terminal time $T$ with the optimal discounted portfolio selection.
\end{lemma}

According to Lemma \ref{th:aux}, we focus on the tractable problem \eqref{model:3} on which the principle of dynamic programming can be applied, and the optimal discounted portfolio selection is explicitly represented in Lemma \ref{le:theta}.

\begin{lemma}[\cite{Zhou2000Continuous}]\label{le:theta}
The optimal discounted mean-variance portfolio selection of problem \eqref{model:3} is given by
\begin{equation}\label{equ:theta0}
	\theta^*_t
	=(-w+w^o+\dfrac{e^{\int_0^T (\frac{\mu(s)-r}{\sigma(s)})^2 ds}}{2\gamma})\dfrac{\mu(t)-r}{\sigma^2(t)}, \quad 0\leqslant t \leqslant T.
\end{equation}
\end{lemma}

Some remarks are in order.
%
%
$\frac{\mu(t)-r}{\sigma(t)}$ is the risk premium which presents the return per unit risk of the risky asset at time $t$.
The optimal portfolio selection is not only related to the risk premium at current time $t$, but also related to the risk premium during the whole planning investment horizon $\int_0^T (\frac{\mu(s)-r}{\sigma(s)})^2 ds$. 
In \eqref{equ:theta0}, the risk premium in the past and future appears as a whole part in the index based on the natural number $e$.
Whenever the risk premium increases, the investor will invest more money in the risky asset.

When considering the financial market with constant return rate and volatility, the optimal discounted pre-commitment portfolio selection is simplified in Lemma \ref{le:constant}.

\begin{lemma}\label{le:constant}
When considering the constant return rate and volatility of risky asset, i.e., $\mu(t)\equiv\mu$, $\sigma(t)\equiv\sigma$, $\forall t\in[0,T]$, the optimal discounted mean-variance portfolio selection \eqref{equ:theta0} can be simplified into
\begin{equation}\label{equ:theta_real_time}
	\theta^*_t
	=(-w+w^o+\dfrac{e^{(\frac{\mu-r}{\sigma})^2\cdot T}}{2\gamma})\dfrac{\mu-r}{\sigma^2}.
\end{equation}
\end{lemma}

\section{The Average and Current Profitability of Risky Asset}\label{section:3}

When implementing the portfolio selection \eqref{equ:theta0},
the investor needs to estimate the model parameters $\mu(t), \sigma(t), \forall t\in[0,T]$ from the historical time series of risky asset price.
The traditional method is to estimate $\mu(t)$ and $\sigma(t)$ with maximum likelihood estimation (MLE) based on the process of discounted risky asset \eqref{equ:S}.
However, in this section, we propose two indexes which present the average profitability of risky asset (AP) and the current probability of risky asset (CP).
And, in the next section, we will show the superiority of the indexes when estimating them instead of $\mu(t), \sigma(t), \forall t\in[0,T]$.

In the process of discounted risky asset \eqref{equ:S}, the risk premium at time $t$ is $\frac{\mu(t)-r}{\sigma(t)}$, $\forall t\in[0,T]$.
If $\frac{\mu(t)-r}{\sigma(t)}\neq0$, the investor can profit from buying or shorting the risky asset. 
Denote by
\begin{equation}\label{equ:A}
	A(t)=\Big(\dfrac{\mu(t)-r}{\sigma(t)}\Big)^2
\end{equation}
CP at time $t$.
The larger value of $A(t)$, the more investor can profit from buying or shorting a share of risky asset.
Then, it follows that,
AP on $[t,T]$ can be denoted by
\begin{equation}\label{equ:K}
	\bold{K}(t,T)= \dfrac{1}{T-t}\int_t^T A(s) ds.
\end{equation}
In this way, the optimal discounted pre-commitment portfolio selection \eqref{equ:theta0} can be represented by AP and CP, which is rewritten as 
\begin{equation}\label{equ:theta1}
	\theta^*_t
	=(-w+w^o+\dfrac{e^{\bold{K}(0,T)\cdot T}}{2\gamma})\dfrac{\sqrt{A(t)}}{\sigma(t)}, \quad 0\leqslant t \leqslant T.
\end{equation}

Thus, at time $t$, what the investor needs to estimate for implementing the portfolio selection \eqref{equ:theta1} is the volatility of risky asset $\sigma(t)$, AP $\bold{K}(0,T)$ and CP $A(t)$.
Then, we introduce a property of $\bold{K}(t,T)$ in Theorem \ref{th:K}, which is important for developing the estimation method of $\bold{K}(0,T)$ in next section.

\begin{theorem}\label{th:K}
Let $\{\widetilde{W}_s\}_{t \leqslant s \leqslant T}$ be the process of discounted wealth under the price of risky asset \eqref{equ:S} and an auxiliary portfolio selection
\begin{equation}\label{equ:aux_theta}
	\widetilde{\theta}_s=\dfrac{\mu(s)-r}{\sigma^2(s)}, \qquad t\leqslant s\leqslant T.
\end{equation}
Then,
\begin{equation}\label{equ:se_order}
	\bold{K}(t,T) =\dfrac{1}{T-t}\Big([\widetilde{W},\widetilde{W}]_{T}-[\widetilde{W},\widetilde{W}]_{t}\Big),
\end{equation}
in which 
$[\widetilde{W},\widetilde{W}]_t$ denotes the second-order variation of $\{\widetilde{W}_s\}$ from time $0$ to $t$, $\forall t\in[0,T]$.
\end{theorem}

\begin{proof}
Choose a partition $\Pi=\{t_0,\dots,t_k\}$ of $[0,t]$, which is a set of times $0=t_0 \leqslant \dots \leqslant t_k=t$. 
We do not require the partition points to be equally spaced, although they are allowed to be. 
The maximum step size of the partition is denoted by $||\Pi||=\max\limits_{i=0,\dots,k-1}(t_{i+1}-t_i)$. 
Then, the second-order variation of the discounted wealth process $\{\widetilde{W}_s\}$ from time $0$ to $t$ is
\begin{equation*}
	[\widetilde{W},\widetilde{W}]_t
	=\lim\limits_{||\Pi||\to0}\sum_{i=0}^{k-1} \Big(\widetilde{W}_{t_{i+1}}-\widetilde{W}_{t_i}\Big)^2.
\end{equation*}
Under the process of risky asset \eqref{equ:S} and the auxiliary portfolio selection \eqref{equ:aux_theta}, the second-order variation becomes
\begin{align*}
	[\widetilde{W},\widetilde{W}]_t
	&=\int_0^t(\widetilde{\theta}_s\sigma(s))^2ds\\
	&=\int_0^t\Big(\dfrac{\mu(s)-r}{\sigma(s)}\Big)^2ds.
\end{align*}
Thus, for $ 0\leqslant t \leqslant T$,
\begin{equation*}
	\bold{K}(t,T) =\dfrac{1}{T-t}\Big([\widetilde{W},\widetilde{W}]_{T}-[\widetilde{W},\widetilde{W}]_{t}\Big).
\end{equation*}
\end{proof}

In fact, \eqref{equ:aux_theta} is the optimal discounted time-consistent mean-variance portfolio selection under the criterion of conditional expectation and variance, i.e.,
\begin{equation}\label{equ:tc}
	{\rm E}_t\Big(W_T\Big)- \dfrac{1}{2} {\rm Var}_t\Big(W_T\Big),
\end{equation}
\citep{Basak2010Dynamic}.
This criterion can yield an analytical equilibrium strategy which is only relevant to the parameters at current time $t$.
Specifically, in \eqref{equ:aux_theta}, the optimal discounted time-consistent mean-variance portfolio selection is relevant to CP.
Thus, the second-order variation on the right side of \eqref{equ:se_order} presents AP which is the benefit of the auxiliary wealth process \eqref{equ:aux_theta}.

\section{Algorithm Design}\label{section:4}

After denoting AP and CP, we now design an algorithm to implement the optimal pre-commitment portfolio selection \eqref{equ:theta1}, without assuming any knowledge about the underlying forms of the return rate and volatility of risky asset $\mu(\cdot)$ and $\sigma(\cdot)$.
In this section, we first summarize the pseudocode for the optimal discounted pre-commitment portfolio selection in Algorithm \ref{alg:para(t)}, 
and then detail the estimations AP $\bold{K}(0,T)$ and CP $A(t)$ in Section \ref{se:K} and \ref{se:A}. 
The superiority of the estimation of $\bold{K}(0,T)$ and $A(t)$ are shown theoretically and numerically, respectively.

To describe the algorithm, we first discretize $[0,T]$ into small equal-length intervals $[t_k,t_{k+1}]$, $k = 0, 1, \dots, N-1$, where $t_0 = 0$ and $t_N = T$.
It is assumed that the price of risky asset can be observed and the portfolio selection can be implemented only at time $t_k$.
At each executable moment $t_k$, Algorithm \ref{alg:para(t)} consists of three successive estimation procedures:
the estimation of volatility of risky asset $\sigma(t_k)$, the estimation of AP $\bold{K}(0,T)$ and the estimation of CP $A(t_k)$.

The estimation of $\sigma(t_k)$ is obtained by traditional MLE \citep{Campbell1996Econometrics},
and the estimation of $\mu(t_k)$ is obtained simultaneously:
\begin{equation}\label{equ:estimate_mu_sigma}
\begin{split}
	&\hat{\mu}(t_k)=\hat{\alpha}+r+\dfrac{\hat{\beta}}{2},\\
	&\hat{\sigma}(t_k)=\sqrt{\hat{\beta}},
\end{split}
\end{equation}
in which
\begin{align*}
	&\hat{\alpha}
	=\dfrac{1}{\Delta t}\dfrac{1}{M+1}\sum_{i=(k-1)-M}^{k-1}\ln{\dfrac{S_{t_{i+1}}}{S_{t_i}}},\\
	&\hat{\beta}
	=\dfrac{1}{\Delta t}\dfrac{1}{M+1}\sum_{i=(k-1)-M}^{k-1}(\ln{\dfrac{S_{t_{i+1}}}{S_{t_i}}}-\hat{\alpha}\Delta t)^2.
\end{align*}

\renewcommand{\algorithmicrequire}{\textbf{Input:}}
\renewcommand{\algorithmicensure}{\textbf{Output:}}
\begin{breakablealgorithm}\label{alg:para(t)}
\caption{The Optimal Pre-commitment Mean-variance Portfolio Selection}
\begin{algorithmic}[1]
	\FOR {$k=0:N-1$}{
		\STATE{Estimate $\sigma(t_k)$ with \eqref{equ:estimate_mu_sigma}.}
		\STATE{Estimate $\bold{K}(0,T)$ with \eqref{equ:estimate_K}.}
		\STATE{Estimate $A(t_k)$ with \eqref{equ:estimate_A}.}
		\STATE{Generate $\theta_{t_k}$ with the optimal discounted mean-variance portfolio selection \eqref{equ:theta1}.}
		\STATE{At next executable moment $t_{k+1}$, obtain the discounted wealth $W_{t_{k+1}}$ by
		\begin{equation*}
			W_{t_{k+1}}=W_{t_k}+\dfrac{S_{t_{k+1}}-S_{t_k}}{S_{t_k}}\theta_{t_k}.
		\end{equation*}}
	}
	\ENDFOR
\end{algorithmic}
\end{breakablealgorithm}

\subsection{The Estimation of AP}\label{se:K}

In this section, we detail the estimation of $\bold{K}(0,T)$ at time $t_k$.
With \eqref{equ:K}, we have 
\begin{equation*}
	\bold{K}(0,T)=\dfrac{\bold{K}(0,t_k)\cdot t_k+\bold{K}(t_k,T)\cdot(T-t_k)}{T}.
\end{equation*}
At time $t_k$, $\bold{K}(0,t_k)$ can be calculated with historical data, but $\bold{K}(t_k,T)$ needs to be estimated.
We consider that AP is period and is about equally in two equal-long periods near each other according to the cyclical pattern of economic activity \citep{Hamilton1989new}.
In this way, $\bold{K}(t_k,T)$ is estimated by $\bold{K}(t_k-(T-t_k), t_k)=\bold{K}(t_{k-(N-k)}, t_k)$.

Following Theorem \ref{th:K}, 
\begin{equation*}
	\bold{K}(t_i,t_j)=\lim\limits_{||\Pi||\to0}\sum_{s=i}^{j-1} \Big(\widetilde{W}_{t_{s+1}}-\widetilde{W}_{t_s}\Big)^2 \cdot \dfrac{1}{t_j-t_i},
	\qquad(t_i,t_j)\in\{(t_0,t_k),(t_{k-(N-k)}, t_k)\},
\end{equation*}
in which $\{\widetilde{W}_{t_s}\}_{i\leqslant s \leqslant j}$ is generated with the auxiliary portfolio selection \eqref{equ:aux_theta}.
Thus, $\bold{K}(0,T)$ can be estimated by
\begin{equation}\label{equ:estimate_K}
	\hat{\bold{K}}(0,T) 
	= \dfrac{1}{T} \Big(\sum_{i=0}^{k-1}+ \sum_{i=k-(N-k)}^{k-1}\Big) \Big(\widetilde{W}_{t_{i+1}}-\widetilde{W}_{t_i}\Big)^2.
\end{equation}
in which $\{\widetilde{W}_{t_s}\}_{i\leqslant s \leqslant j}$ is generated with
\begin{equation*}
	\widetilde{\theta}_{t_s}=\dfrac{\hat{\mu}(t_s)-r}{\hat{\sigma}^2(t_s)}.
\end{equation*}

In \eqref{equ:estimate_K}, $\hat{\bold{K}}(0,T)$ fully considers the time-dependence characteristics of risk premium.
In fact,
\begin{align}
	\hat{\bold{K}}(0,T) 
	&=\dfrac{1}{T} \Big(\sum_{i=0}^{k-1}+ \sum_{i=k-(N-k)}^{k-1}\Big) \Big(\widetilde{W}_{t_{i+1}}-\widetilde{W}_{t_i}\Big)^2 \notag\\
	&\approx \dfrac{1}{T} \Big(\sum_{i=0}^{k-1}+ \sum_{i=k-(N-k)}^{k-1}\Big) (\widetilde{\theta}_{t_i}\sigma(t_i))^2 \cdot \Delta t \notag\\
	&= \dfrac{1}{T} \Big(\sum_{i=0}^{k-1}+ \sum_{i=k-(N-k)}^{k-1}\Big) \Big(\dfrac{\hat{\mu}(t_i)-r}{\hat{\sigma}^2(t_i)}\Big)^2 \sigma^2(t_i) \cdot \Delta t \notag\\
	&= \dfrac{1}{N} \Big(\sum_{i=0}^{k-1}+ \sum_{i=k-(N-k)}^{k-1}\Big) \Big(\dfrac{\hat{\mu}(t_i)-r}{\hat{\sigma}(t_i)}\Big)^2 \dfrac{\sigma^2(t_i)}{\hat{\sigma}^2(t_i)}\label{equ:ss},
\end{align}
which implicitly includes the true value of volatility of risky asset $\sigma(t_i)$.
%
%
We show in Theorem \ref{th:relative_error} that, for $(\frac{\mu(t_k)-r}{\sigma(t_k)})^2$, $(\frac{\hat{\mu}(t_k)-r}{\hat{\sigma}(t_k)})^2 \frac{\sigma^2(t_k)}{\hat{\sigma}^2(t_k)}$ in \eqref{equ:ss} is more exactly than $(\frac{\hat{\mu}(t_k)-r}{\hat{\sigma}(t_k)})^2$ with lower absolute error.

\begin{theorem}\label{th:relative_error}
Denote the true value of return rate and volatility of risky asset at time $t_k$ as $\mu(t_k), \sigma(t_k)$, and their estimated value as $\hat{\mu}(t_k), \hat{\sigma}(t_k)$.
Consider the assumption that $\sigma^2(t_k)<\hat{\sigma}^2(t_k)$ and $2(\frac{\mu(t_k)-r}{\sigma(t_k)})^2<(\frac{\hat{\mu}(t_k)-r}{\hat{\sigma}(t_k)})^2$,
we have
\begin{equation*}
	\Big|(\frac{\hat{\mu}(t_k)-r}{\hat{\sigma}(t_k)})^2 \frac{\sigma^2(t_k)}{\hat{\sigma}^2(t_k)}-(\frac{\mu(t_k)-r}{\sigma(t_k)})^2\Big|
	<\Big|(\frac{\hat{\mu}(t_k)-r}{\hat{\sigma}(t_k)})^2-(\frac{\mu(t_k)-r}{\sigma(t_k)})^2\Big|.
\end{equation*}
\end{theorem}

\begin{proof}
The assumption $2(\frac{\mu(t_k)-r}{\sigma(t_k)})^2<(\frac{\hat{\mu}(t_k)-r}{\hat{\sigma}(t_k)})^2$ implies
$(\frac{\hat{\mu}(t_k)-r}{\hat{\sigma}(t_k)})^2 (\frac{\sigma(t_k)}{\mu(t_k)-r})^2>2$.
And
\begin{equation*}
	\Big(\dfrac{\hat{\mu}(t_k)-r}{\hat{\sigma}(t_k)}\Big)^2 \Big(\dfrac{\sigma(t_k)}{\mu(t_k)-r}\Big)^2
	\Big(\dfrac{\sigma^2(t_k)}{\hat{\sigma}^2(t_k)}+1\Big)
	>2\cdot 1
	>2.
\end{equation*}
Further, according to the assumption that $\sigma^2(t_k)<\hat{\sigma}^2(t_k)$, we have
\begin{equation*}
	\Big(\dfrac{\hat{\mu}(t_k)-r}{\hat{\sigma}(t_k)}\Big)^2 \Big(\dfrac{\sigma(t_k)}{\mu(t_k)-r}\Big)^2
	\Big(\Big(\dfrac{\hat{\mu}(t_k)-r}{\hat{\sigma}(t_k)}\Big)^2 \Big(\dfrac{\sigma(t_k)}{\mu(t_k)-r}\Big)^2
	\Big(\dfrac{\sigma^2(t_k)}{\hat{\sigma}^2(t_k)}+1\Big)-2\Big)
	\Big(\dfrac{\sigma^2(t_k)}{\hat{\sigma}^2(t_k)}-1\Big)
	<0.
\end{equation*}
Thus,
\begin{gather*}
	\Big(\Big(\dfrac{\hat{\mu}(t_k)-r}{\hat{\sigma}(t_k)}\Big)^2 \dfrac{\sigma^2(t_k)}{\hat{\sigma}^2(t_k)}\Big(\dfrac{\sigma(t_k)}{\mu(t_k)-r}\Big)^2-1\Big)^2
	< \Big(\Big(\dfrac{\hat{\mu}(t_k)-r}{\hat{\sigma}(t_k)}\Big)^2 \Big(\dfrac{\sigma(t_k)}{\mu(t_k)-r}\Big)^2-1\Big)^2\\
	\Big|\Big(\dfrac{\hat{\mu}(t_k)-r}{\hat{\sigma}(t_k)}\Big)^2 \dfrac{\sigma^2(t_k)}{\hat{\sigma}^2(t_k)}-\Big(\dfrac{\mu(t_k)-r}{\sigma(t_k)}\Big)^2\Big|
	<\Big|\dfrac{\Big(\hat{\mu}(t_k)-r}{\hat{\sigma}(t_k)}\Big)^2-\Big(\dfrac{\mu(t_k)-r}{\sigma(t_k)}\Big)^2\Big|.
\end{gather*}
\end{proof}

\subsection{The Estimation of CP}\label{se:A}

In this section, we detail the estimation of $A(t_k)$.
Due to the superiority of the estimation of $\bold{K}(0,T)$, we estimate $A(t_k)$ with the $t$-time estimation of $\bold{K}(0,T)$:
\begin{equation}\label{equ:estimate_A}
\begin{split}
	\hat{A}(t_k)	
	&=\hat{\bold{K}}(0,T)\\
	&= \dfrac{1}{T} \Big(\sum_{i=0}^{k-1} \Big(\widetilde{W}_{t_{i+1}}-\widetilde{W}_{t_i}\Big)^2 \cdot \Delta t
	+ \sum_{i=k-(N-k)}^{k-1} \Big(\widetilde{W}_{t_{i+1}}-\widetilde{W}_{t_i}\Big)^2 \cdot \Delta t\Big).
\end{split}
\end{equation}
It is well known that, in MLE, nothing is gained in terms of estimation accuracy of the return rate $\mu(t_k)$ by choosing finer observation intervals \citep{MERTON1980323}.
And, the estimation error in $\mu(t_k)$ tends to have great influence in portfolio selection \citep{MJ1991Sensitivity}.
In the following of this section, we numerically show that the estimation in \eqref{equ:estimate_A} does reduce the estimation standard deviation of $\sqrt{A(t_k)}$ when choosing finer observation intervals.

Given $\mu^*(t_k), \sigma^*(t_k)$ and $S_{t_k}$, we independently generate 10,000 samples for discounted price of risky asset at next observable time $\{S_{t_{k+1}}^i\}_{i=1}^{10000}$ with the process \eqref{equ:S}.
In this setting, the risk premium $\frac{\mu(t_k)-r}{\sigma(t_k)}$ which is presented in the optimal pre-commitment portfolio selection \eqref{equ:theta1} as $\sqrt{A(t_k)}$ can be estimated.
In Table \ref{table:mu}, two ways for estimating the risk premium at time $t_k$ are compared:

a -- Estimate the risk premium $\frac{\mu(t_k)-r}{\sigma(t_k)}$ with $\sqrt{\hat{A}(t_k)}$, in which $\hat{A}(t_k)$ is estimated with \eqref{equ:estimate_A}.

b -- Estimate the risk premium $\frac{\mu(t_k)-r}{\sigma(t_k)}$ with $\frac{\hat{\mu}(t_k)-r}{\hat{\sigma}(t_k)}$, in which $\hat{\mu}(t_k), \hat{\sigma}(t_k)$ are estimated with traditional MLE \eqref{equ:estimate_mu_sigma}.

Table \ref{table:mu} shows the estimation standard deviation of $\sqrt{A(t_k)}$ under different observation intervals $\Delta t$.
The price of risky asset is observed monthly, weekly and daily when $\Delta t=21/252, 12/252$ and $1/252$.
In Table \ref{table:mu}, we find that, as $\Delta t$ decreases, the estimation with $\sqrt{A(t_k)}$ in \eqref{equ:estimate_A} is more stable with lower standard deviation, 
while the traditional MLE of $\mu(t_k), \sigma(t_k)$ does not.

\setlength\LTcapwidth{\textwidth}
\begin{longtable}{p{3.2cm}|p{1.5cm}|p{1.5cm}|p{1.5cm}|p{1.5cm}|p{1.5cm}|p{1.5cm}}
\caption{The standard deviation of the estimation of risk premium.}\label{table:mu}\\ 
\toprule
\thead{$\sigma^*(t_k)\equiv0.1$}
&\multicolumn{2}{|c}{$\mu^*(t_k)\equiv0.08$}
&\multicolumn{2}{|c}{$\mu^*(t_k)\equiv0.1$}
&\multicolumn{2}{|c}{$\mu^*(t_k)\equiv0.12$}
  \\\hline        
standard deviation
&\thead{$a$} &\thead{$b$} 
&\thead{$a$} &\thead{$b$} 
&\thead{$a$} &\thead{$b$}
  \\\hline        
  \thead{$\Delta t=21/252$}
  &\thead{$1.0324$}
  &\thead{$1.1553$}
  &\thead{$1.1384$}
  &\thead{$1.1613$}
  &\thead{$1.2711$}
  &\thead{$1.1697$}
  \\\hline
  \thead{$\Delta t=12/252$}
  &\thead{$0.7508$}
  &\thead{$1.0803$}
  &\thead{$0.8266$}
  &\thead{$1.0834$}
  &\thead{$0.9169$}
  &\thead{$1.0877$}
  \\\hline
  \thead{$\Delta t=1/252$}
  &\thead{$0.5193$}
  &\thead{$1.0013$}
  &\thead{$0.5682$}
  &\thead{$1.0014$}
  &\thead{$0.6224$}
  &\thead{$1.0016$}
  \\
\bottomrule
\end{longtable}

\section{Numerical Study}\label{section:5}

We compare the performance of mean-variance portfolio selection A, B, N and T:

A -- The portfolio selection \eqref{equ:theta1} which is based on estimating AP according to the second-order variation of an auxiliary wealth process with \eqref{equ:estimate_K} and estimating CP with \eqref{equ:estimate_A}.

B -- The portfolio selection \eqref{equ:theta1} which is based on estimating AP with the consideration that the risk premium is constant, and estimating CP with the MLE of return rate and volatility of risky asset. 

N -- The so-called ``buy-and-hold'' strategy, which invests all the money into the risky asset at the initial time 0 and makes no adjustments for optimization in $[0,T]$.

T -- The portfolio selection \eqref{equ:theta1} with actual values of return rate and volatility of risky asset.

\noindent
The portfolio selection A, B, T are different in estimating the model parameters when implementing the analytical solution \eqref{equ:theta1}.
The portfolio selection A is the strategy shown in our automatio Algorithm \ref{alg:para(t)}.
The portfolio selection B is the strategy implemented with traditional MLE described in Section \ref{section:4}.
The portfolio selection T is the strategy with the true values of parameters $\mu(t), \sigma(t), \forall 0\leqslant t \leqslant T$. 
Portfolio selection T is expected to perform the best, although the true values of parameters are unknown in practice and the performance is not attainable.
The portfolio selection N displays the performance of the well-diversified stock index itself, which does not involve any optimization or estimation and is only shown as a counterpart.

We consider the risk aversion coefficient $\gamma=1.4$, the risk-free interest rate $r=0.02$, the planning investment horizon $T = 1$ year and the initial fund $w^o = 1$.
The planning investment horizon is discretized into $N=252$ sub-intervals, and the discretization length $\Delta t=\frac{1}{252}$ with each subinterval representing one trading day in financial market.
%
Each portfolio selection starts with an initial estimation of parameters and updates the estimation through time based on observed data. 
We set the window for estimating the parameters to be $M=252$ data points.
In the following of this section, we conduct comparisons of portfolio selection A, B, N and T in various simulation settings, 
including the GBM case where the return rate and volatility of risky asset are constant, 
the Heston case where the return rate and volatility of risky are influenced by an exogenous variable,
and the real financial market case.

\subsection{The Geometric Brownian Motion Case}\label{se:gbm}

In this section, the data of discounted risky asset price is generated by GBM with constant return rate and volatility
\begin{equation*}
	dS_t=(\mu-r) S_tdt+\sigma S_tdB_t.
\end{equation*}
Reasonable values of $\mu$ and $\sigma$ are taken from the set $\{0.08,0.1,0.12\}$ and $\{0.1\}$, which are ``typical" for simulating a well-diversified stock index.
In the GBM case with constant $\mu, \sigma$, the optimal discounted mean-variance portfolio selection \eqref{equ:theta1} can be simplified into \eqref{equ:theta_real_time}.
Thus, there is no difference between portfolio selection A and B except the estimation error of CP at each trading day.

We independently generate 10,000 processes of risky asset price under a given parameters setting $\{\mu,\sigma\}$.
And, the mean estimation of CP for each process is shown in subfigure (a)-(c) of Figure \ref{fig:rho_wT} with boxplot graphs under different portfolio selections.
In subfigure (a)-(c), the estimation of CP in portfolio selection A is more stable with lower estimating volatility, which has been incompletely shown in Section \ref{se:A}.
In other words, the estimation of CP in portfolio selection A is more accurate than that in portfolio selection B in most of the 10,000 processes.
In this way, the investment result in portfolio selection A can perform better.

Subfigure (d)-(f) of Figure \ref{fig:rho_wT} show the terminal wealth of portfolio selection A, B compared to the portfolio selection T under these 10,000 processes. 
In subfigure (d)-(f), the distributions of terminal wealth from portfolio selection A are shifted to the right compared to those from portfolio selection B, which means portfolio selection A yields a higher terminal wealth.
What's more, the distributions of terminal wealth from portfolio selection A are close to portfolio selection T, which validates the practical merit of our portfolio selection A.

\begin{figure}[H]
\begin{minipage}[t]{.48\linewidth}
\includegraphics [width=1\textwidth] {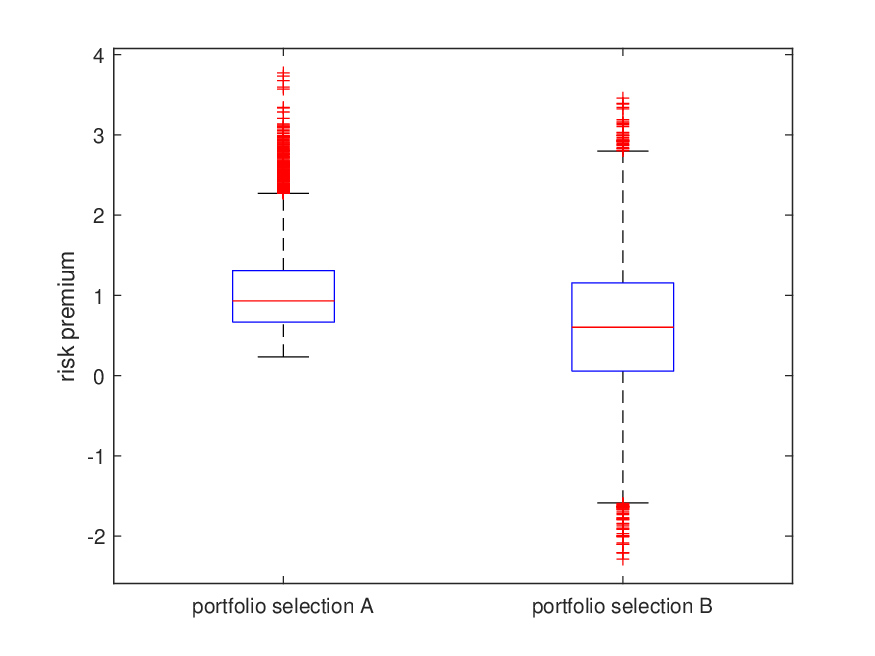}
\centerline{\small (a)~$\mu=0.08$}
\end{minipage}
\begin{minipage}[t]{.48\linewidth}
\includegraphics [width=1\textwidth] {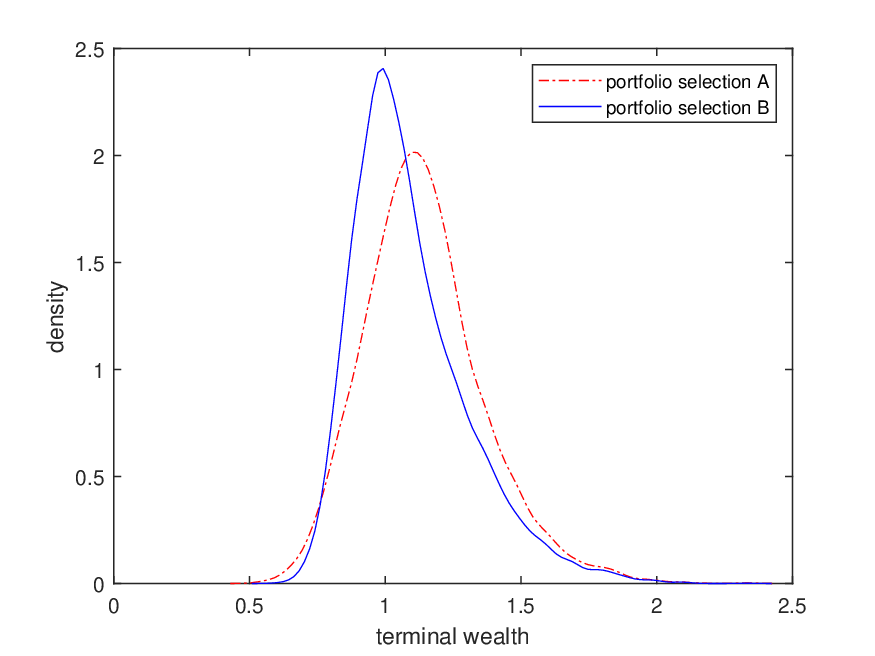}
\centerline{\small (d)~$\mu=0.08$}
\end{minipage}
\end{figure}

\begin{figure}[H]
\begin{minipage}[t]{.48\linewidth}
\includegraphics [width=1\textwidth] {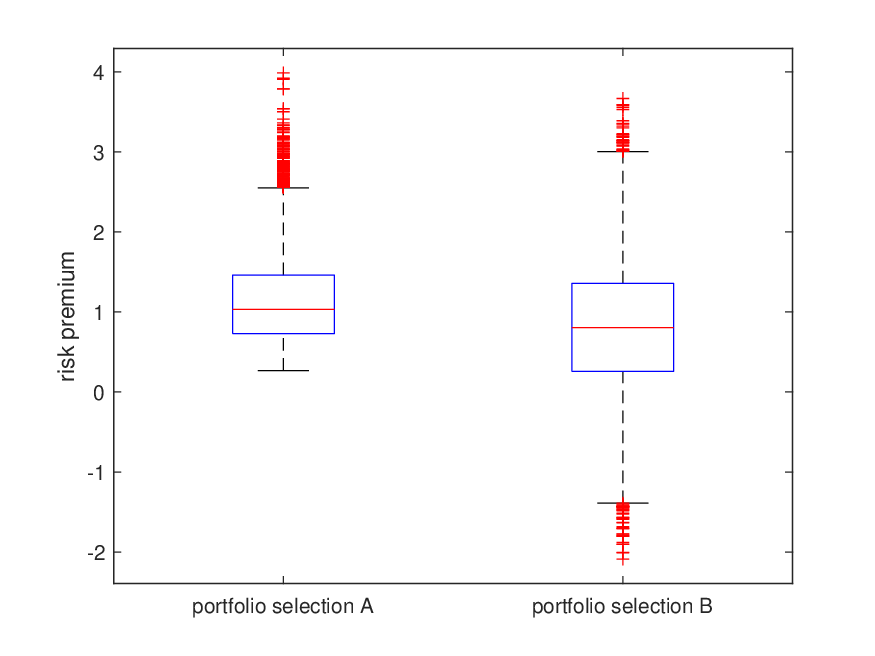}
\centerline{\small (b)~$\mu=0.1$}
\end{minipage}
\begin{minipage}[t]{.48\linewidth}
\includegraphics [width=1\textwidth] {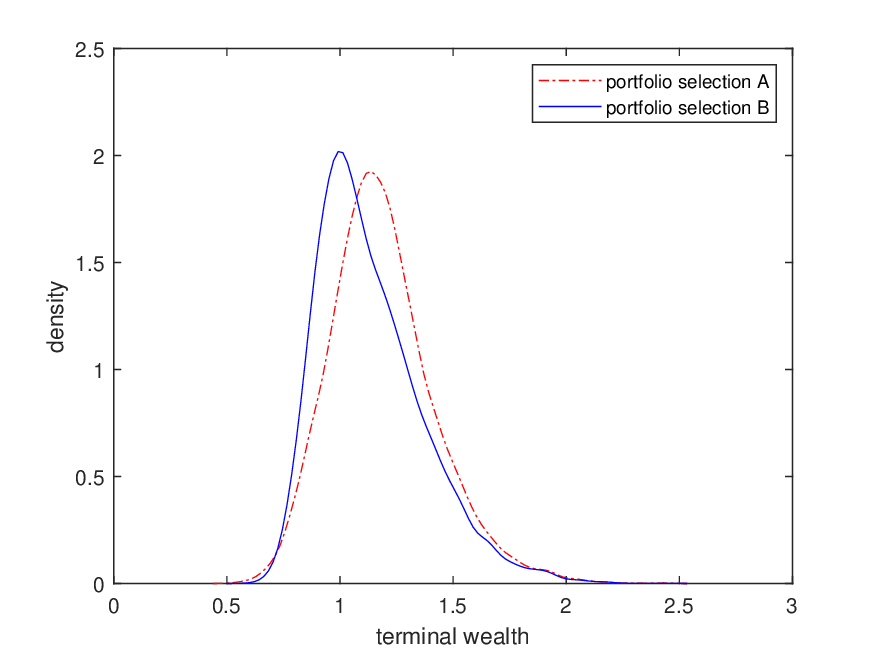}
\centerline{\small (e)~$\mu=0.1$}
\end{minipage}
\end{figure}

\begin{figure}[H]
\begin{minipage}[t]{.48\linewidth}
\includegraphics [width=1\textwidth] {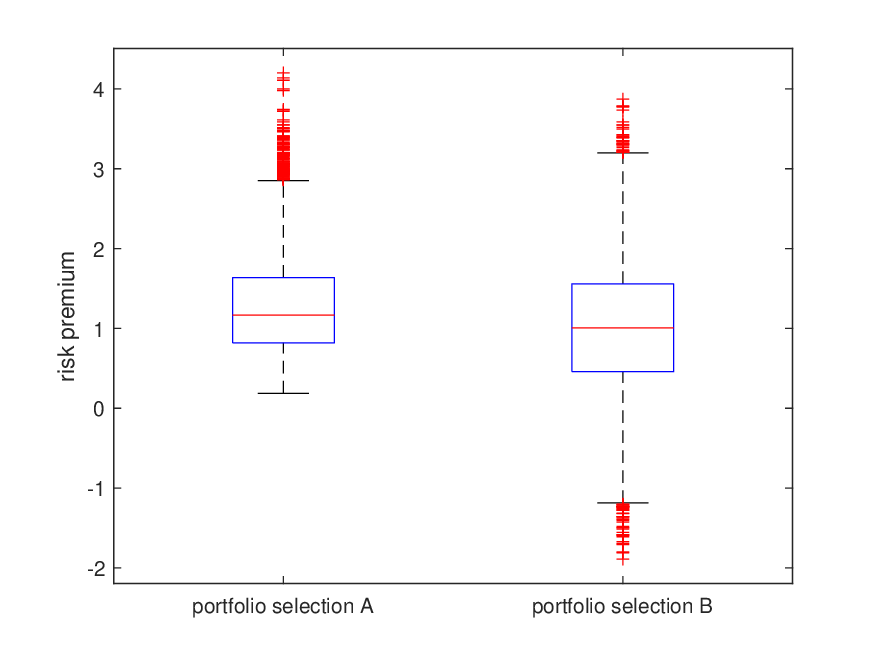}
\centerline{\small (c)~$\mu=0.12$}
\end{minipage}
\begin{minipage}[t]{.48\linewidth}
\includegraphics [width=1\textwidth] {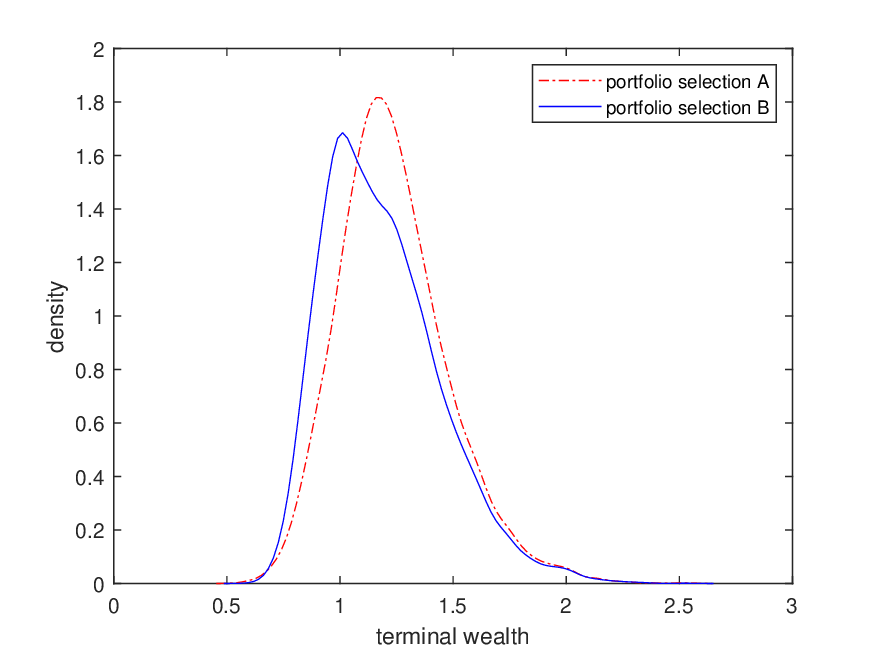}
\centerline{\small (f)~$\mu=0.12$}
\end{minipage}
\caption{The estimation of CP and corresponding terminal wealth.}
\label{fig:rho_wT}
\end{figure}

\subsection{The Heston Case}

The simulation studies in the previous subsection confirm the outperformance of our portfolio selection A under GBM with constant return rate and volatility.
To further test the practical feasibility of portfolio selection A, we consider the Heston model \citep{Heston1993closed} for the discounted risky asset price:
\begin{equation}\label{equ:heston_model}
\begin{split}
	&dS_t=a X_t S_t dt + \sqrt{X_t} S_t dB_t^1\\
	&dX_t=\iota (k-X_t) dt + v\sqrt{X_t} (\kappa dB_t^1+\sqrt{1-\kappa^2} dB_t^2),
\end{split}
\end{equation}
in which the Brownian motion $dB_t^1, dB_t^1$ are independent and $\mu(t), \sigma(t)$ are related to both time $t$ and exogenous variable $X_t$.
The exogenous variable $X_t$ could represent some economic factors, such as the dividend yields or the instantaneous volatility \citep{Dai2023Learning}.
In Heston case, $\mu(t)=aX_t+r, \sigma(t)=\sqrt{X_t}$, which change frequently over time.

Before comparing portfolio selection A, B, N and T under various evaluation criteria, we first illustrate their performance in estimating AP.
We take parameters setting $a=8.5$, $\iota=42.5$, $k=0.01$, $v=0.6$, $X_0=0.02$, $\kappa=-0.7$, which are “typical” for simulation purpose \citep{Yan2019Open}.
Then, we take samples to simulate the process of risky asset price with the parameters setting and the estimations of AP are correspondingly shown in Figure \ref{fig:K}.

In each subfigure of Figure \ref{fig:K}, 
the dot dash line is the estimation of AP in portfolio selection A with \eqref{equ:estimate_K}, while the solid line is the traditional MLE of AP in portfolio selection B.
And, the estimations of AP are updated through time based on observed data.
Under the actual parameters setting, the theoretically value of AP in portfolio selection T is shown with dashed line.
We find that, in all the samples, the estimations of AP in portfolio selection A are more accurate than that in portfolio selection B and are very close to the corresponding theoretically value, which illustrate the practical merit of our portfolio selection A.
What's more, the estimation of AP in portfolio selection A changes slowly over time, being 
\begin{equation*}
	\hat{\bold{K}}(0,T) 
	=\dfrac{1}{T} \Big(\sum_{i=0}^{k-1}+ \sum_{i=k-(N-k)}^{k-1}\Big) \Big(\widetilde{W}_{t_{i+1}}-\widetilde{W}_{t_i}\Big)^2
\end{equation*}
at time $t_k$ and
\begin{equation*}
	\hat{\bold{K}}(0,T) 
	=\dfrac{1}{T} \Big(\sum_{i=0}^{k}+ \sum_{i=(k+1)-(N-(k+1))}^{k}\Big) \Big(\widetilde{W}_{t_{i+1}}-\widetilde{W}_{t_i}\Big)^2
\end{equation*}
at time $t_{k+1}$.
However, AP in portfolio selection B is only related to the estimation of risk premium at current, which is neither accurate nor stable in Heston case.

\begin{figure}[H]
\begin{minipage}[t]{.32\linewidth}
\includegraphics [width=1\textwidth] {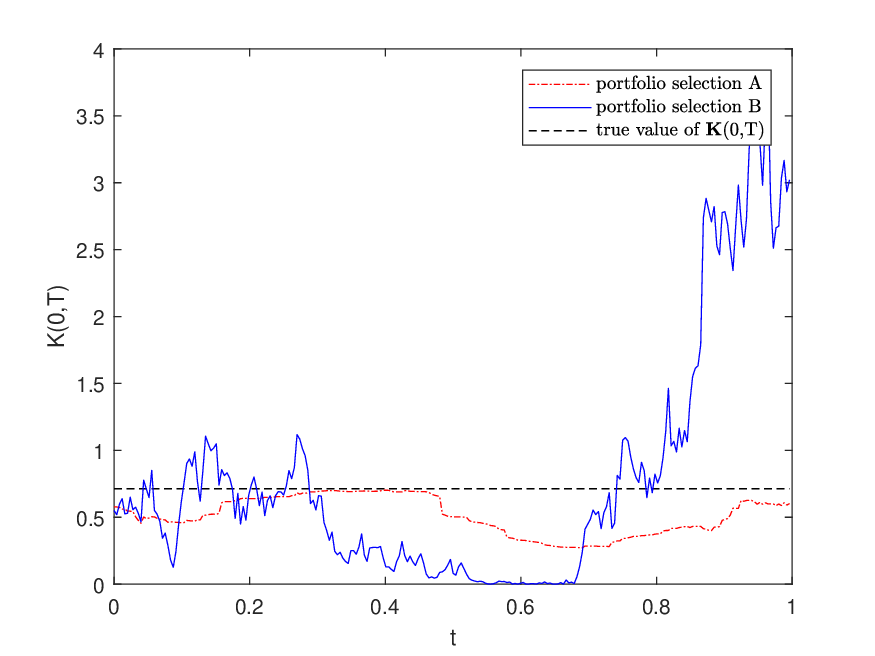}
\end{minipage}
\begin{minipage}[t]{.32\linewidth}
\includegraphics [width=1\textwidth] {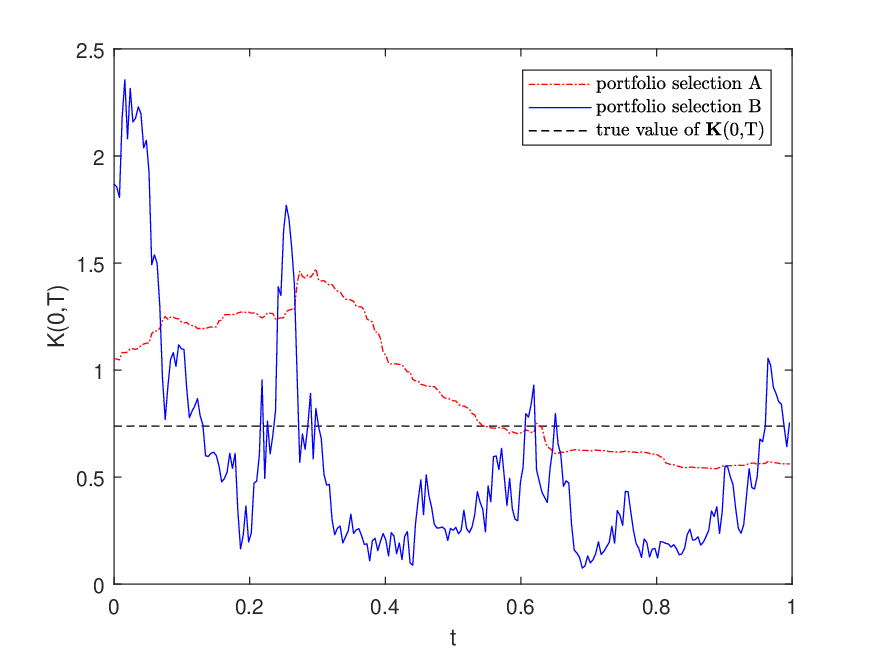}
\end{minipage}
\begin{minipage}[t]{.32\linewidth}
\includegraphics [width=1\textwidth] {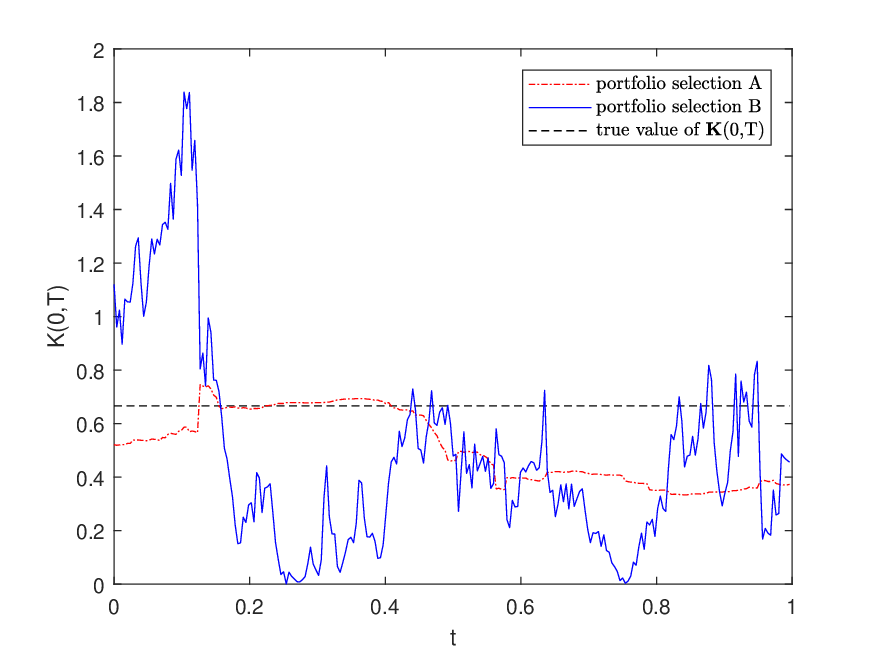}
\end{minipage}
\end{figure}

\begin{figure}[H]
\begin{minipage}[t]{.32\linewidth}
\includegraphics [width=1\textwidth] {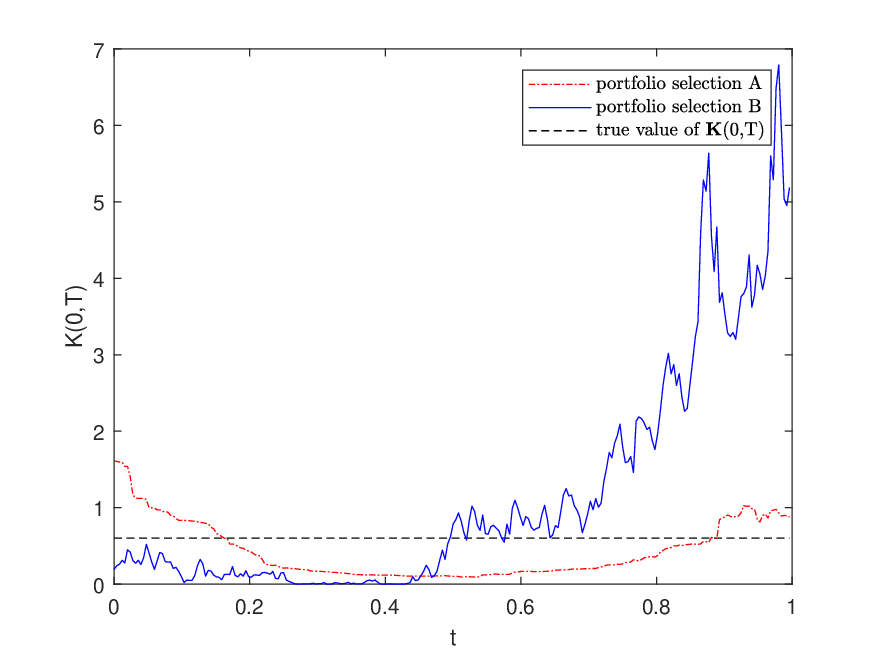}
\end{minipage}
\begin{minipage}[t]{.32\linewidth}
\includegraphics [width=1\textwidth] {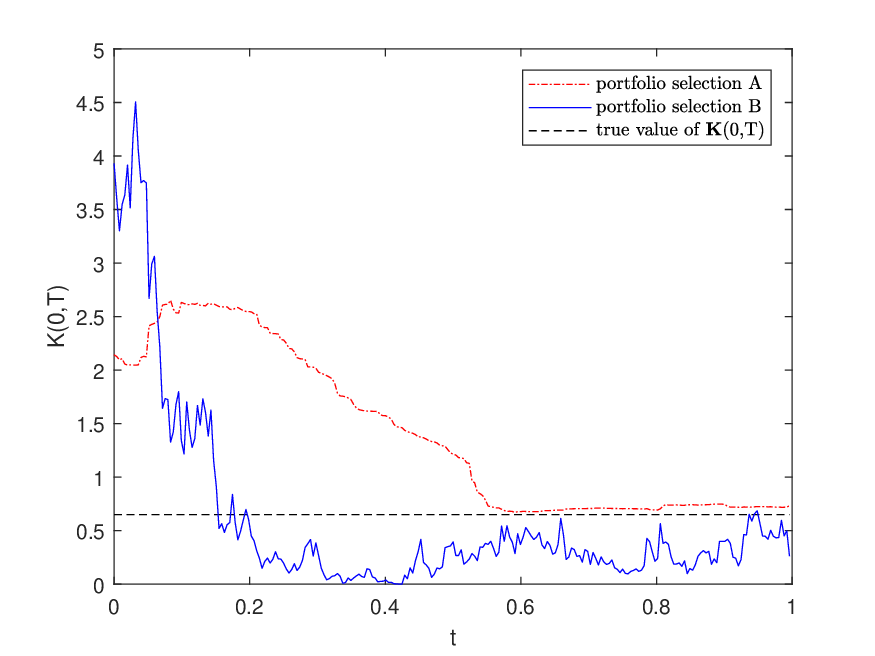}
\end{minipage}
\begin{minipage}[t]{.32\linewidth}
\includegraphics [width=1\textwidth] {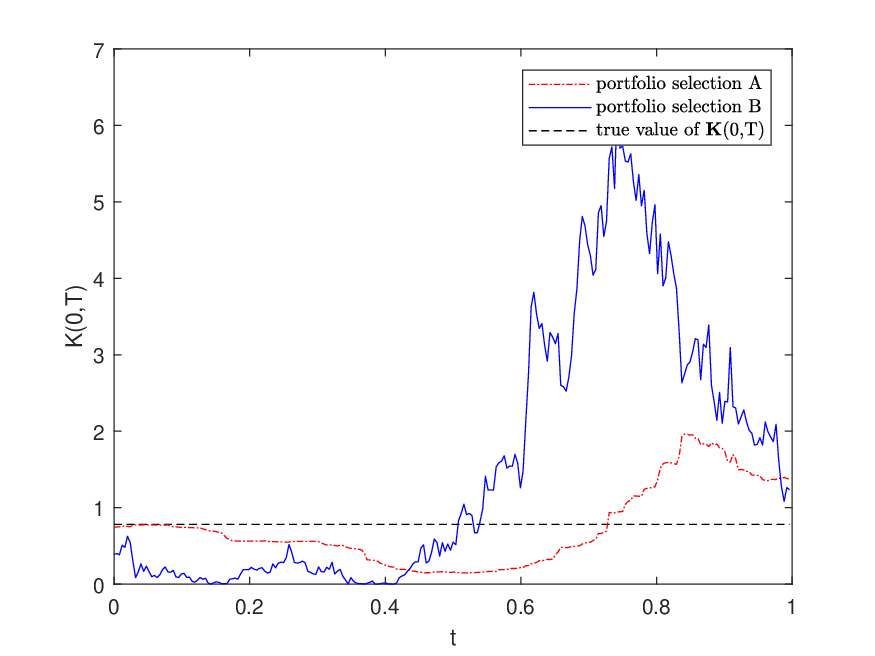}
\end{minipage}
\caption{The process of risky asset price and the corresponding estimation of $\bold{K}(0,T)$ in the simulated financial markets with $a=8.5, \iota=42.5, k=0.01, v=0.6, \kappa=-0.7, X_0=0.02$.}
\label{fig:K}
\end{figure}

Further, we conduct comparisons of portfolio selection A, B, N and T under the Heston model with evaluation criteria Certainty-Equivalent Return (CEQ) \citep{DeMiguel2007Optimal}
\begin{equation*}
	\text{CEQ}={\rm E}\Big(\dfrac{W_T-W_0}{W_0}\Big)-\gamma {\rm Var}\Big(\dfrac{W_T-W_0}{W_0}\Big)
\end{equation*}
and Sharpe Ratio (SR) \citep{Sharpe1994The}
\begin{equation*}
	\text{SR}=\dfrac{{\rm E}\Big(\dfrac{W_T-W_0}{W_0}\Big)-r}{{\rm Std} \Big(\dfrac{W_T-W_0}{W_0}\Big)}.
\end{equation*}
CEQ and SR can represent investment return per unit risk of the portfolio selections.
In particular, when $W_0=w^o=1$, the objective function can be expressed as CEQ plus 1.
In this way, CEQ can measure the value of objective function under parameters estimated error when implementing different portfolio selections.
What's more, we also introduce Turnover Rate (TR) 
\begin{equation*}
	\text{TR}
	={\rm E}\Big(\sum_{i=1}^{N}|\dfrac{\theta_{(i-1)\frac{T}{N}}}{W_{i\frac{T}{N}}}\dfrac{S_{i\frac{T}{N}}}{S_{(i-1)\frac{T}{N}}}-\dfrac{\theta_{i\frac{T}{N}}}{W_{i\frac{T}{N}}}|\Big)
\end{equation*}
to show the transition cost of different portfolio selections.

We implement the portfolio selections derived based on the process \eqref{equ:S}, but test them on data from Heston model.
We independently generate 10,000 processes of risky asset price under a given parameters setting $\{a,\iota,k,v,X_0,\kappa\}$. 
After implementing each portfolio selection, 10,000 processes of wealth are obtained. 
Then, CEQ and SR is calculated with the expectation and variance of the 10,000 one-year investment return. 
TR is calculated with the average of absolute value of the trades across the risky asset and risk-free asset.
To assess the robustness of the outperformance of portfolio selection A, we simulate the portfolio selections with different setting $\iota\in\{40,42.5,45\}, \kappa\in\{-0.6,-0.7,-0.8\}$.
$\iota$ presents the mean reversion rate of exogenous variable $X_t$,
and $\kappa$ affects the correlation between the process of risky asset and exogenous variable.

CEQ under different simulated markets is shown in Table \ref{table:heston_cer1},
in which the portfolio selection A does perform better than B and N in all the markets with significant difference $(p<0.001)$.
When $\iota=0, v=0$, the Heston case \eqref{equ:heston_model} degenerates into GBM with constant parameters $\mu, \sigma$, which is discussed in Section \ref{se:gbm}.
When $\iota, v$ are far from 0, the risky asset return deviates from normal distribution.
And, the portfolio selection A still outperforms B with higher CEQ when the financial market model is assumed incorrectly, i.e., the process \eqref{equ:S} vs Heston Case.

%

We report SR of portfolio selection A, B, N and T under the Heston case in Table \ref{table:heston_sr1}.
Although SR is not the optimization target of mean-variance analysis, SR of portfolio selection A is higher than B in all the market setting with significant level p<0.001.
And the gap of SR between portfolio selection A and N is slight with the relative differences less than 6\%.

TR measures the transaction cost of portfolio selections, which is shown in Table \ref{table:heston_tr1}.
In particular, portfolio selection N does not involve any adjustment of the composition of risky asset basket during the whole planning investment horizon, and TR of portfolio selection N is 0.
In Table \ref{table:heston_tr1}, TR of portfolio selection A is lower than B with lower transition cost, which further proves the superiority of portfolio selection A.

\setlength\LTcapwidth{\textwidth}
\begin{longtable}{p{4cm}|p{2cm}|p{2cm}|p{2cm}|p{2cm}}
\caption{Comparison of CEQ for different portfolio selections in the Heston model with $a=8.5, k=0.01, v=0.6, X_0=0.02$.}\label{table:heston_cer1}\\ 
\toprule
~\quad\qquad$\{\iota,\kappa\}$   &\qquad A &\qquad B &\qquad N &\qquad T
\\\hline        
  \thead{$\{40,-0.6\}$}
  &\thead{$0.1273$}
  &\thead{$0.0667$}
  &\thead{$0.0927$}
  &\thead{$0.1334$}
  \\\hline
  \thead{$\{40,-0.7\}$}
  &\thead{$0.1288$}
  &\thead{$0.0679$}
  &\thead{$0.0930$}
  &\thead{$0.1348$}
  \\\hline
  \thead{$\{40,-0.8\}$}
  &\thead{$0.1306$}
  &\thead{$0.0696$}
  &\thead{$0.0934$}
  &\thead{$0.1364$}
  \\\hline
  \thead{$\{42.5,-0.6\}$}
  &\thead{$0.1269$}
  &\thead{$0.0665$}
  &\thead{$0.0926$}
  &\thead{$0.1331$}
  \\\hline
  \thead{$\{42.5,-0.7\}$}
  &\thead{$0.1284$}
  &\thead{$0.0676$}
  &\thead{$0.0929$}
  &\thead{$0.1344$}
  \\\hline
  \thead{$\{42.5,-0.8\}$}
  &\thead{$0.1301$}
  &\thead{$0.0692$}
  &\thead{$0.0933$}
  &\thead{$0.1359$}
  \\\hline
  \thead{$\{45,-0.6\}$}
  &\thead{$0.1265$}
  &\thead{$0.0663$}
  &\thead{$0.0926$}
  &\thead{$0.1327$}
  \\\hline
  \thead{$\{45,-0.7\}$}
  &\thead{$0.1279$}
  &\thead{$0.0674$}
  &\thead{$0.0928$}
  &\thead{$0.1340$}
  \\\hline
  \thead{$\{45,-0.8\}$}
  &\thead{$0.1295$}
  &\thead{$0.0690$}
  &\thead{$0.0932$}
  &\thead{$0.1354$}
  \\
\bottomrule
\end{longtable}

\setlength\LTcapwidth{\textwidth}
\begin{longtable}{p{4cm}|p{2cm}|p{2cm}|p{2cm}|p{2cm}}
\caption{Comparison of SR for different portfolio selections in the Heston model with $a=8.5, k=0.01, v=0.6, X_0=0.02$.}\label{table:heston_sr1}\\ 
\toprule
~\quad\qquad$\{\iota,\kappa\}$   &\qquad A &\qquad B &\qquad N &\qquad T
\\\hline        
  \thead{$\{40,-0.6\}$}
  &\thead{$0.8000$}
  &\thead{$0.5235$}
  &\thead{$0.8467$}
  &\thead{$0.8325$}
  \\\hline
  \thead{$\{40,-0.7\}$}
  &\thead{$0.8110$}
  &\thead{$0.5276$}
  &\thead{$0.8580$}
  &\thead{$0.8432$}
  \\\hline
  \thead{$\{40,-0.8\}$}
  &\thead{$0.8239$}
  &\thead{$0.5339$}
  &\thead{$0.8707$}
  &\thead{$0.8553$}
  \\\hline
  \thead{$\{42.5,-0.6\}$}
  &\thead{$0.7971$}
  &\thead{$0.5229$}
  &\thead{$0.8434$}
  &\thead{$0.8297$}
  \\\hline
  \thead{$\{42.5,-0.7\}$}
  &\thead{$0.8074$}
  &\thead{$0.5268$}
  &\thead{$0.8540$}
  &\thead{$0.8397$}
  \\\hline
  \thead{$\{42.5,-0.8\}$}
  &\thead{$0.8195$}
  &\thead{$0.5328$}
  &\thead{$0.8658$}
  &\thead{$0.8511$}
  \\\hline
  \thead{$\{45,-0.6\}$}
  &\thead{$0.7945$}
  &\thead{$0.5223$}
  &\thead{$0.8405$}
  &\thead{$0.8273$}
  \\\hline
  \thead{$\{45,-0.7\}$}
  &\thead{$0.8042$}
  &\thead{$0.5261$}
  &\thead{$0.8504$}
  &\thead{$0.8367$}
  \\\hline
  \thead{$\{45,-0.8\}$}
  &\thead{$0.8156$}
  &\thead{$0.5319$}
  &\thead{$0.8615$}
  &\thead{$0.8473$}
  \\
\bottomrule
\end{longtable}

\setlength\LTcapwidth{\textwidth}
\begin{longtable}{p{4cm}|p{2cm}|p{2cm}|p{2cm}|p{2cm}}
\caption{Comparison of TR for different portfolio selections in the Heston model with $a=8.5, k=0.01, v=0.6, X_0=0.02$.}\label{table:heston_tr1}\\ 
\toprule
~\quad\qquad$\{\iota,\kappa\}$   &\qquad A &\qquad B &\qquad N &\qquad T
\\\hline        
  \thead{$\{40,-0.6\}$}
  &\thead{$3.4673$}
  &\thead{$17.8306$}
  &\thead{$0.0000$}
  &\thead{$3.0495$}
  \\\hline
  \thead{$\{40,-0.7\}$}
  &\thead{$3.4713$}
  &\thead{$17.7008$}
  &\thead{$0.0000$}
  &\thead{$3.0455$}
  \\\hline
  \thead{$\{40,-0.8\}$}
  &\thead{$3.4799$}
  &\thead{$17.5420$}
  &\thead{$0.0000$}
  &\thead{$3.0510$}
  \\\hline
  \thead{$\{42.5,-0.6\}$}
  &\thead{$3.4745$}
  &\thead{$17.8739$}
  &\thead{$0.0000$}
  &\thead{$3.0581$}
  \\\hline
  \thead{$\{42.5,-0.7\}$}
  &\thead{$3.4759$}
  &\thead{$17.7531$}
  &\thead{$0.0000$}
  &\thead{$3.0548$}
  \\\hline
  \thead{$\{42.5,-0.8\}$}
  &\thead{$3.4848$}
  &\thead{$17.6002$}
  &\thead{$0.0000$}
  &\thead{$3.0606$}
  \\\hline
  \thead{$\{45,-0.6\}$}
  &\thead{$3.4798$}
  &\thead{$17.9112$}
  &\thead{$0.0000$}
  &\thead{$3.0662$}
  \\\hline
  \thead{$\{45,-0.7\}$}
  &\thead{$3.4806$}
  &\thead{$17.7988$}
  &\thead{$0.0000$}
  &\thead{$3.0632$}
  \\\hline
  \thead{$\{45,-0.8\}$}
  &\thead{$3.4891$}
  &\thead{$17.6498$}
  &\thead{$0.0000$}
  &\thead{$3.0688$}
  \\
\bottomrule
\end{longtable}

\subsection{The Real Financial Market Case}

We study the dynamic allocation between the risk-free asset and one of the stock indexes, NASDAQ, S\&P500, DJI, to illustrate the performance of our portfolio selection A in the real financial market.
In the real financial market, the true values of parameters can not be obtained, and we only compare the portfolio selection A, B and N.
We take data of stock indexes from January 2014 to December 2023.



There are 1,760 one-year planning investment horizons from January 2014 to December 2023.
CEQ of portfolio selections is calculated with the expectation and variance of the 1,760 one-year investment return which is shown in Table \ref{table:real_cer}.
In Table \ref{table:real_cer}, portfolio selection A does perform better than B, but not as well as portfolio selection N.
Thus, we make further adjustment to portfolio selection A with the key idea that the sophisticated portfolio selection can be combined with portfolio selection N to improve performance \citep{Tu2011Markowitz}. 
More precisely, we propose a simple rule to combine portfolio selection A and N:
\begin{equation*}
\theta^{\text{A+N}}(t)
=
\begin{cases}
	\theta^{\text{A}}(t), & \text{when}~\hat{\sigma}(t)<0.1,\\
	\theta^{\text{N}}(t), & \text{otherwise}.
\end{cases}
\end{equation*}
And the CEQ of adjusted portfolio selection A (denoted as A+N), which is shown in Table \ref{table:real_cer}, outperforms N.
What's more, the SR of portfolio selection A+N performs best among the portfolio selection A+N, B, N in all the stock indexes.
And, the TR of portfolio selection A+N still remains a low value.
We leave the optimal combination method between portfolio selection A and N for future study.

\setlength\LTcapwidth{\textwidth}
\begin{longtable}{p{4cm}|p{2cm}|p{2cm}|p{2cm}|p{2cm}}
\caption{Comparison of CEQ for different portfolio selections in the real financial market.}\label{table:real_cer}\\ 
\toprule
  &~\quad A+N &\qquad A &\qquad B &\qquad N
  \\\hline        
  \thead{NASDAQ}
  &\thead{$0.1139$}
  &\thead{$0.0288$}
  &\thead{$-0.0407$}
  &\thead{$0.1032$}
  \\\hline
  \thead{S\&P500}
  &\thead{$0.1022$}
  &\thead{$0.0365$}
  &\thead{$-0.0622$}
  &\thead{$0.0879$}
  \\\hline
  \thead{DJI}
  &\thead{$0.0877$}
  &\thead{$0.0239$}
  &\thead{$-0.0625$}
  &\thead{$0.0791$}
  \\
\bottomrule
\end{longtable}
  
\setlength\LTcapwidth{\textwidth}
\begin{longtable}{p{4cm}|p{2cm}|p{2cm}|p{2cm}|p{2cm}}
\caption{Comparison of SR for different portfolio selections in the real financial market.}\label{table:real_sr}\\ 
\toprule
  &~\quad A+N &\qquad A &\qquad B &\qquad N
  \\\hline        
  \thead{NASDAQ}
  &\thead{$0.7351$}
  &\thead{$0.5038$}
  &\thead{$0.2678$}
  &\thead{$0.6888$}
  \\\hline
  \thead{S\&P500}
  &\thead{$0.7467$}
  &\thead{$0.3832$}
  &\thead{$-0.0048$}
  &\thead{$0.6732$}
  \\\hline
  \thead{DJI}
  &\thead{$0.6430$}
  &\thead{$0.3314$}
  &\thead{$-0.0224$}
  &\thead{$0.6267$}
  \\
\bottomrule
\end{longtable}

\setlength\LTcapwidth{\textwidth}
\begin{longtable}{p{4cm}|p{2cm}|p{2cm}|p{2cm}|p{2cm}}
\caption{Comparison of TR for different portfolio selections in the real financial market.}\label{table:real_tr}\\ 
\toprule
  &~\quad A+N &\qquad A &\qquad B &\qquad N
  \\\hline        
  \thead{NASDAQ}
  &\thead{$0.6651$}
  &\thead{$6.3882$}
  &\thead{$22.3844$}
  &\thead{$0.0000$}
  \\\hline
  \thead{S\&P500}
  &\thead{$0.8350$}
  &\thead{$5.3954$}
  &\thead{$20.0998$}
  &\thead{$0.0000$}
  \\\hline
  \thead{DJI}
  &\thead{$0.5572$}
  &\thead{$5.3863$}
  &\thead{$21.1776$}
  &\thead{$0.0000$}
  \\
\bottomrule
\end{longtable}

\section{Conclusion}\label{section:6}

We consider the mean-variance portfolio selection in the Geometric Brownian Motion (GBM) with time-varying return rate and volatility of risky asset.
We define two indexes which presents the average profitability of the risky asset (AP) and the current profitability of the risky asset (CP).
Then the pre-commitment mean-variance portfolio selection can be represented by AP and CP, 
and the unknown future information of financial market can be separated with index AP.
For estimating AP, we propose an auxiliary wealth process which corresponds to time-consistent mean-variance portfolio selection and estimate AP with the second-order variation of the auxiliary wealth process.
For CP, we approximate CP with the estimation of AP.
Our numerical study gives strong support to the implement of portfolio selection which is based on the estimations of AP and CP.

\bibliography{MV_tc-16-01}

\end{document}